\documentclass[letterpaper, 10 pt, conference]{ieeeconf}
\IEEEoverridecommandlockouts

\usepackage{amssymb,amsmath}

\usepackage{amsthm}
\usepackage{mathtools}
\usepackage{multirow}
\usepackage{graphicx}
\usepackage{comment}
\usepackage[sort,compress,noadjust]{cite}
\usepackage[font=footnotesize]{subcaption}
\usepackage[font=footnotesize]{caption}

\usepackage{amsfonts}
 \usepackage{tabularx}
\usepackage{graphicx}
\usepackage{array}
\usepackage{float}
\usepackage{hhline}
\usepackage{algorithmic}
\usepackage[linesnumbered, ruled, vlined]{algorithm2e}
\usepackage{setspace}
\usepackage{colortbl}
\usepackage{arydshln}
\SetCommentSty{mycommfont}

\usepackage{accents}

\usepackage{placeins}

\theoremstyle{plain}
\newtheorem{theorem}{Theorem}
\theoremstyle{plain}
\newtheorem{proposition}{Proposition}
\theoremstyle{plain}
\newtheorem{lemma}{Lemma}
\theoremstyle{plain}
\newtheorem{corollary}{Corollary}
\theoremstyle{definition}

\theoremstyle{remark}

\newtheoremstyle{definition}{}{}{}{}{\bfseries}{.}{.5em}{\thmname{#1}\thmnumber{ #2}\thmnote{ (#3)}}
\theoremstyle{definition}

\usepackage[colorlinks,urlcolor=blue]{hyperref}
\usepackage[capitalize]{cleveref}
\crefformat{equation}{(#2#1#3)}
\Crefformat{equation}{Equation~(#2#1#3)}
\Crefname{equation}{Equation}{Eqs.}

\title{\LARGE\bf Analysis of the Geometric Heat Flow Equation: \\ 
Computing Geodesics in Real-Time with Convergence Guarantees}

\author{Samuel G. Gessow and Brett T. Lopez%
\thanks{
All authors are with VECTR Laboratory, University of California, Los Angeles, Los Angeles, CA, USA, {\tt\footnotesize \{sgessow,  btlopez\}@ucla.edu}}
}

\begin{document}
\maketitle
\thispagestyle{empty}


\begin{abstract}
    We present an analysis on the convergence properties of the so-called geometric heat flow equation for computing geodesics (extremal curves) on Riemannian manifolds.
    Computing geodesics numerically in real time has become an important capability across several fields, including control and motion planning. 
    The geometric heat flow equation involves solving a parabolic partial differential equation whose solution is a geodesic.
    In practice, solving this PDE numerically can be done efficiently, and tends to be more numerically stable and exhibit a better rate of convergence compared to numerical optimization. 
    We prove that the geometric heat flow equation is exponentially stable in $L_2$ if the curvature of the Riemannian manifold does not exceed a positive bound and that asymptotic convergence in $L_2$ is always guaranteed.
    We also present a pseudospectral method that leverages Chebyshev polynomials to accurately compute geodesics in only a few milliseconds for non-contrived manifolds.
    Our analysis was verified with our custom pseudospectral method by computing geodesics on common non-Euclidean surfaces, and in feedback for a contraction-based controller with a non-flat metric for a nonlinear system. 
\end{abstract}

\section{Introduction}
\label{sec:introduction}

Finding the shortest path between two points on a non-Euclidean manifold has become an important aspect in control, motion planning, computer graphics, and various other fields. 
Computing the shortest path (or curve) on a Riemannian manifold, i.e., a smooth manifold equipped with a smooth spatially varying inner product, is formally stated as finding the extremal curves of the arc length functional, and these extremal curves are called geodesics \cite{do1992riemannian}.
Various methods have been proposed for computing geodesics numerically depending on how the Riemannian manifold is represented (continuous or discrete) and whether the shortest point-to-point or point-to-many path is desired.
Two common methods for computing point-to-point geodesics on continuous manifolds are gradient descent \cite{leung2017nonlinear,manchester2017control} and the geometric heat flow method \cite{jost2005riemannian,belabbas2017new,liu2019affine}. 
For gradient descent, the problem is formulated as a two-point boundary value problem where the Riemannian energy functional is minimized. 
Using a set of basis functions to represent the solution, e.g., Chebyshev polynomials, the optimization can be posed over the finite-dimensional space of unknown coefficients.
This method is often used with contraction-based feedback controllers \cite{leung2017nonlinear,manchester2017control,lopez2020adaptive,singh2023robust}, but the computational demand and no provable convergence rate guarantee has motivated researchers to explore other methods, e.g., see \cite{wang2020continuous,tsukamoto2021contraction}.

An alternative approach relies on solving a parabolic partial differential equation that deforms a curve with fixed boundary conditions until it becomes an extremal curve {for a given homotopy class}. 
The parabolic PDE is known as the \emph{geometric heat flow equation} \cite{jost2005riemannian} because it resembles the heat equation from thermodynamics. 
The key advantage of the geometric heat flow method over gradient descent is that solving the PDE can be done very efficiently by recasting the PDE as an initial value problem for a coupled system of ordinary differential equations.
In other words, a solution to the PDE can be obtained simply by numerically integrating a coupled system of ODEs given initial conditions until a convergence criterion is met. 
This is formally known as the method of lines \cite{schiesser2012numerical}, and has been shown to be effective at computing extremal trajectories for non-holonomic \cite{belabbas2017new} and non-affine \cite{liu2019affine}.
For high dimensional systems, the method of lines can be combined with pseudospectral methods to improve scalability and solution speed \cite{adu2025bring}.
Despite promising empirical results concerning numerical stability and convergence rate, a comprehensive analysis of the geometric heat flow method has not yet been conducted.  

The main contribution of this paper is the stability analysis of the geometric heat flow equation.
We show that any initial curve will exponentially converge in $L_2$ to a geodesic if the curvature of the Riemannian manifold does not exceed a certain positive value (which we specify).
Our approach mirrors the derivation of the so-called Jacobi field in Riemannian geometry \cite{do1992riemannian}, but we instead obtain a parabolic PDE with a source term related to the Riemannian curvature.
We then employ PDE stability theory \cite{krstic2008boundary} to prove exponential stability in $L_2$ depending on the Riemannian curvature; asymptotic stability in $L_2$ always holds. 
In other words, any curve will converge exponentially (or asymptotically depending on the curvature) in $L_2$ to a geodesic by numerically solving the geometric heat flow PDE.
Our analysis is verified numerically with a custom pseudospectral method that efficiently computes geodesics within a few milliseconds for non-contrived manifolds. 
Numerical tests were conducted by computing geodesics on classical 2D surfaces (sphere, torus, and egg box), and the method was used in feedback with a contraction-based controller and a non-flat metric to stabilize a third-order nonlinear system \cite{manchester2017control}.

\textit{Notation:} 
The set of positive and strictly-positive scalars is $\mathbb{R}_{+}$ and $\mathbb{R}_{>0}$.
The set of real symmetric matrices is $\mathcal{S}$, and the set of symmetric positive definite matrices is $\mathcal{S}_+$.
The partial derivative of a multivariable function is $\partial_s c(s,\tau) = \partial c / \partial s$.
The general inner product for two vectors $u,w $ is $\langle u, w\rangle_g = \sum_{i,j} g_{ij} u_i w_j$, where the subscript $g$ is omitted when the inner product is the Euclidean norm.
The $L_2$ norm of a function $f$ is $\|f(s)\|^2_{L_2[a,b]} = \int_a^b \|f(s)\|^2 ds$.

\section{Riemannian Geometry Preliminaries}
Let $\mathcal{M}$ be a smooth manifold and let the tangent space at a point $p \in \mathcal{M}$, denoted as $T_p\mathcal{M}$, be the set of all tangent vectors at $p$.
If a smooth manifold is equipped with a smoothly varying inner product on the tangent space $T_p \mathcal{M}$, i.e., there exists $g_p : T_p \mathcal{M} \times T_p \mathcal{M} \rightarrow \mathbb{R} : (u,w) \mapsto \langle u, w \rangle_{g(p)} = \sum_{i,j} g_{ij}(p) u_i w_j$, then $(\mathcal{M},g)$ is a Riemannian manifold and $g$ is the Riemannian metric. 
In some works, the Riemannian metric is represented as a tensor $G : \mathcal{M} \rightarrow \mathcal{S}_+$ where the $i,j$ element of $G$ is $g_{ij}$.
The Riemannian metric $g$ defines local geometric notions such as angles, length, and orthogonality at every point ${p} \in \mathcal{M}$.

Let ${p},\,q \in \mathcal{M}$ and ${c}:[0,1] \rightarrow \mathcal{M} : s \mapsto c(s)$ be a regular (i.e., $\partial_s c\neq 0 ~  \forall s \in [0,1]$) parametrized differentiable curve such that ${c}(0)={p}$ and ${c}(1)=q$.
The length $\mathcal{L}$ and Riemannian energy $\mathcal{E}$ functionals of curve ${c}(s)$ are
\begin{equation*}
    \begin{aligned}
        \mathcal{L}(c(s)) & =\int_0^1 \sqrt{\langle \partial_s c, \partial_s c \rangle}_{g(c)}ds \\ 
        \mathcal{E}(c(s)) &= \frac{1}{2} \int_0^1{\langle \partial_s c, \partial_s c \rangle_{g(c)}}ds.
    \end{aligned}
\end{equation*}
Let $\Xi({p},q)$ denote the family of regular curves with ${c}(0)=p$ and ${c}(1)=q$.
The Riemannian distance between $p$ and $q$ is $d(p,q) = \inf_{c(s)\in\Xi}\mathcal{L}({c}(s))$.
By the Hopf-Rinow theorem, under suitable conditions a minimizing curve known as a \textit{minimum geodesic} ${\gamma}:[0,1] \rightarrow \mathcal{M}$ is guaranteed to exist with the unique property $\mathcal{E}({\gamma}(s))=\frac12\mathcal{L}({\gamma(s)})^2 \leq \frac12\mathcal{L}({c}(s))^2 \leq \mathcal{E}({c}(s))$.
In other words, a minimum geodesic is the shortest curve between two points on a manifold.
Geodesics can be characterized using the covariant derivative, which is a generalization of the directional derivative for vector fields defined along a curve. 
When the Riemannian manifold $(\mathcal{M},g)$ is embedded as surface $\mathcal{F} \subset \mathbb{R}^d$ in a higher-dimensional Euclidean space (which is always possible via the Nash embedding theorems), the covariant derivative can be interpreted as the projection of the standard derivative onto the tangent plane of the surface $\mathcal{F}$ at a given point. 
With this embedding in mind\footnote{The discussion uses the ambient Euclidean space only as a means of providing intuition about core concepts; it is well known that Riemannian geometry can be developed abstractly without an ambient space \cite{do1992riemannian}.}, several important properties of the covariant derivative can be deduced. 
The first is compatibility of the metric: for a regular curve $c(s)$ and vector fields $U, \, W : [0,1] \rightarrow T_{c(s)} \mathcal{M} : s \mapsto U(s), \, W(s)$, we have 
\begin{equation*}
    \frac{d}{ds} \langle U, W \rangle_{g(c(s))} = \langle \frac{D}{ds} U, W \rangle_{g(c(s))} + \langle U, \frac{D}{ds} W \rangle_{g(c(s))},
\end{equation*}
where we denote the covariant derivative as the operator $D/ds$ that projects the standard derivative onto the surface $\mathcal{F}$. 
Notationally, this operator is sometimes written as $\nabla_{\partial_s c}$ to explicitly show the differentiation is along the vector field $\partial_s c \in T_{c(s)} \mathcal{M}$.
We will present a coordinate description of the covariant derivative in \cref{sec:implementation}, but it is worth mentioning now that the covariant derivative captures the evolution of the components \emph{and} the coordinate system of a vector field as one moves along $c(s)$. 

The second important property is that the covariant derivative provides a necessary and sufficient condition for a curve to be a geodesic.
Specifically, if $\gamma(s)$ is a geodesic, then \cite{do1992riemannian}
\begin{equation*}
    \frac{D}{ds} \partial_s \gamma = \nabla_{\partial_s \gamma} \partial_s \gamma = 0,
\end{equation*}
which essentially states that ``speed'' of a geodesic, i.e., $\langle \partial_s \gamma, \partial_s \gamma \rangle_{g(\gamma)}$ is constant, or that the vector field $\partial_s \gamma$ is parallel transported along $\gamma(s)$.
Note that any minimum geodesic is a geodesic, but a geodesic may not be minimal in the sense of Riemannian distance because there can be several extremal curves that satisfy $D (\partial_s \gamma) / d s = 0$.
We will omit the ``minimal'' qualifier when discussing geodesics.

The final important property is that the covariant derivative is symmetric, i.e., torsion-free, so it commutes with the standard derivative. 
In other words, given a curve $c: [0,1] \times \mathbb{R}_+ \rightarrow \mathcal{M} : (s,\tau) \mapsto c(s,\tau)$, then $D (\partial_s c) / d \tau = D (\partial_\tau c) / d s$.
However, the covariant derivative does not necessarily commute with itself: for $c(s,\tau)$ and smooth vector field $W : [0,1] \rightarrow  T_{c(s,\tau)} \mathcal{M}$, one has \cite{do1992riemannian}
\begin{equation*}
    \frac{D}{d \tau} \frac{D}{d s} W - \frac{D}{d s} \frac{D}{d \tau} W = R(\partial_\tau c, \partial_s c) W
\end{equation*}
where $R : T \mathcal{M} \times T \mathcal{M} \times T \mathcal{M} \rightarrow T \mathcal{M}$ is the Riemannian curvature tensor that measures the non-commutativity of the covariant derivative. 
Later, a more intuitive measure of curvature will be discussed.

\section{Main Results}
\label{sec:main_result}
As discussed previously, there are several ways to compute geodesics given two points on a Riemannian manifold $(\mathcal{M},g)$.
The method of interest in this work involves the so-called geometric heat flow equation \cite{jost2005riemannian}
\begin{equation}
    \label{eq:heat_flow}
    \partial_\tau c = \alpha \frac{D}{d s} \partial_s c
\end{equation}
where $c: [0,1] \times \mathbb{R}_+ \rightarrow \mathcal{M} : (s,\tau) \rightarrow c(s,\tau)$ is a parameterized regular curve which has an additional dummy variable $\tau$, $D / ds$ is the covariant derivative, and $\alpha \in \mathbb{R}_{>0}$.
Formally, $c(s,\tau)$ is restricted to be a proper variation of a geodesic $\gamma(s)$ where $c(0,\tau) = \gamma(0)$ and $c(1,\tau) = \gamma(1)$.
From a dynamical systems perspective, \cref{eq:heat_flow} can be viewed as a gradient flow for the curve $c(s,\tau)$.
If $c(s,\tau)$ were a geodesic, then, by definition, $D (\partial_s c) / ds = 0 \implies \partial_\tau c = 0$, so the ``equilibrium points'' of \cref{eq:heat_flow} are geodesics.
The natural question then arises as to whether \cref{eq:heat_flow} is a stable PDE in the $L_2$ sense where any curve $c(s,\tau)$ that connects two points on $\mathcal{M}$ will converge in $L_2$ to a geodesic $\gamma(s)$ that connects the same two points.
The following proposition establishes the baseline convergence behavior of \cref{eq:heat_flow}.

\begin{proposition}
    \label{prop:heat_flow}
    Any regular curve connecting two points on the Riemannian manifold $(\mathcal{M},g)$ that satisfies the geometric heat flow equation \cref{eq:heat_flow} will asymptotically converge in $L_2$ to a geodesic connecting the same two points on $(\mathcal{M},g)$.  
\end{proposition}

\begin{proof}
    The proof loosely follows from \cite{jost2005riemannian}.
    Let $c : [0,1] \times  \mathbb{R}_+ : (s,\tau) \mapsto c(s,\tau)$ be a regular curve connecting two points on the Riemannian manifold ($\mathcal{M},g$). 
    Moreover, let $c(s,\tau)$ be a proper variation of the geodesic $\gamma(s)$.
    In other words, the curve $c(s,\tau)$ and geodesic $\gamma(s)$ must coincide at their endpoints for all $\tau \in \mathbb{R}_{>0}$, so $\partial_\tau c(0,\tau) = \partial_\tau c(1,\tau) = 0$.
    By definition, the Riemannian energy of curve $c(s,\tau)$ is
    \begin{equation*}
        \mathcal{E}(c(s,\tau)) = \frac{1}{2} \int_0^1 \langle\partial_s c, \partial_s c \rangle_{g(c)} ds.
    \end{equation*}
    With a slight abuse of notation, we will denote the Riemannian energy as a functional of $c(\tau)$, i.e., $\mathcal{E} = \mathcal{E}(c(\tau))$, since the integral essentially eliminates the $s$ dependency.
    Differentiating the Riemannian energy with respect to $\tau$ gives
    \begin{equation*}
        \begin{aligned}
            \frac{d}{d \tau}\mathcal{E}(c(\tau)) & = \frac{1}{2} \int_0^1 \frac{d}{d \tau} \langle \partial_s c, \partial_s c \rangle_{g(c)} ds \\
            & = \int_0^1 \langle \partial_s c, \frac{D}{d \tau} \partial_s c\rangle_{g(c)} ds \\ 
            & = \int_0^1 \langle \partial_s c, \frac{D}{d s} \partial_\tau c\rangle_{g(c)} ds,
        \end{aligned}
    \end{equation*} 
    where we make use of the metric compatibility property of the covariant derivative to eliminate the term associated with differentiating the metric, and the third line follows from symmetry, i.e., $D(\partial_s c) / d\tau = D(\partial_\tau c) / ds$. 
    Noting that
    \begin{equation*}
        \begin{aligned}
            \frac{d}{ds} \int_0^1 \langle \partial_s c, \partial_\tau c \rangle_{g(c)} ds  = &  \int_0^1 \langle \frac{D}{ds} \partial_s c, \partial_\tau c \rangle_{g(c)} ds  \\ 
            & + \int_0^1 \langle \partial_s c, \frac{D}{d s} \partial_\tau c \rangle_{g(c)} ds,            
        \end{aligned}
    \end{equation*}
    then $d\mathcal{E}(c(\tau)) / d \tau$ becomes
    \begin{equation*}
        \begin{aligned}
        \frac{d}{d \tau} \mathcal{E}(c(\tau)) & = \langle \partial_s c, \partial_\tau c \rangle_{g(c)} |^{s=1}_{s=0}  - \int_0^1 \langle \frac{D}{d s} \partial_s c, \partial_\tau c \rangle_{g(c)} ds \\ 
        & = - \int_0^1 \langle \frac{D}{d s} \partial_s c, \partial_\tau c \rangle_{g(c)} ds,
        \end{aligned}
    \end{equation*}
    where the inner product evaluation vanishes because $c(s,\tau)$ is a proper variation, so $\partial_\tau c(0,\tau) = \partial_\tau c(1,\tau) = 0$.
    Since $c(s,\tau)$ satisfies the geometric heat flow equation \cref{eq:heat_flow}, then
    \begin{equation*}
        \frac{d}{d \tau} \mathcal{E}(c(\tau)) = - \alpha \int_0^1 \langle \frac{D}{d s} \partial_s c, \frac{D}{ds} \partial_s c \rangle_{g(c)} ds,
    \end{equation*}
    which implies $d\mathcal{E}(c(\tau))/d \tau$ is negative semidefinite. 
    Since $\mathcal{E}(c(\tau))$ is positive definite and \cref{eq:heat_flow} is a parabolic PDE with strong smoothness properties, LaSalle's invariance principle \cite{luo2012stability} can be used to assert that the curve $c(s,\tau)$ will converge to the largest invariant set corresponding to $d \mathcal{E}(c(\tau)) / d \tau = 0$. 
    From above, $d \mathcal{E}(c(\tau)) / d \tau = 0 \iff D(\partial_s c) / ds = 0$, which is the criterion for a curve to be a geodesic.
    Therefore, $c(s,\tau) \rightarrow \gamma(s)$ in $L_2$ as $\tau \rightarrow \infty$.
\end{proof}

The system \cref{eq:heat_flow} being asymptotically stable in $L_2$ is an important result because it guarantees that any arbitrary curve will eventually converge in $L_2$ to a geodesic by solving \cref{eq:heat_flow}.  
However, this result can be strengthened to exponential convergence under certain conditions related to the curvature of the Riemannian manifold $(\mathcal{M},g)$.
To show this, we start with the following proposition.

\begin{proposition}
    \label{prop:second_order_cov}
    Let $c: [0,1] \times  \mathbb{R}_+ \rightarrow \mathcal{M} : (s,\tau) \mapsto c(s,\tau)$ be a regular curve that satisfies the geometric heat flow equation \cref{eq:heat_flow}.
    Then $c(s,\tau)$ also satisfies     
    \begin{equation}
        \label{eq:second_order_cov}
        \frac{1}{\alpha} \frac{D}{d \tau} \partial_\tau c = \frac{D^2}{d s^2} \partial_\tau c + R (\partial_\tau c, \partial_s c) \, \partial_s c,
    \end{equation}
    where the operators $D / d \tau$ and $D / d s$ are the covariant derivatives along the curve $c(s,\tau)$ and $R : T \mathcal{M} \times T \mathcal{M} \times T \mathcal{M} \rightarrow T \mathcal{M}$ is the Riemannian curvature tensor.
\end{proposition}
\begin{proof}
    Suppose $c(s,\tau)$ satisfies the geometric heat flow equation \cref{eq:heat_flow}.  
    Since $\partial_\tau c \in T \mathcal{M}$, differentiating \cref{eq:heat_flow} requires the covariant derivative because the resulting vector field must still live in the tangent space of $\mathcal{M}$.
    Taking the covariant derivative of \cref{eq:heat_flow} with respect to $\tau$ gives
    \begin{equation*}
        \frac{1}{\alpha} \frac{D}{d \tau} \partial_\tau c = \frac{D}{d\tau} \frac{D}{d s} \partial_s c.
    \end{equation*}
    The order of differentiation on the right-hand side of the above expression can be reversed by recalling the definition of the Riemannian curvature tensor, namely for a curve $c(s,\tau)$ and vector field $W$ we have \cite{do1992riemannian}
    \begin{equation*}
        R(\partial_\tau c, \partial_s c) W = \frac{D}{d \tau} \frac{D}{d s} W - \frac{D}{d s} \frac{D}{d \tau} W,
    \end{equation*}
    so by setting $W = \partial_s c$ we get
    \begin{equation*}
        \frac{1}{\alpha} \frac{D}{d \tau} \partial_\tau c =  \frac{D}{d s} \frac{D}{d\tau} \partial_s c + R(\partial_\tau c, \partial_s c) \partial_s c.
    \end{equation*}
    Finally, by symmetry, $D (\partial_\tau c) / ds = D (\partial_s c) / d \tau $, so
    \begin{equation*}
        \frac{1}{\alpha} \frac{D}{d \tau} \partial_\tau c = \frac{D^2}{d s^2} \partial_\tau c + R(\partial_\tau c, \partial_s c) \partial_s c. \qedhere
    \end{equation*}
\end{proof}

\begin{corollary}
    Let $J : [0,1] \times \mathbb{R}_+ \rightarrow T\mathcal{M}$ be the vector field defined as $J(s,\tau) = \partial_\tau c$.
    By selecting an orthonormal basis parallel to $c(s,\tau)$, \cref{eq:second_order_cov} becomes
    \begin{equation}
        \label{eq:jacobi_heat}
        \frac{1}{\alpha} \frac{D}{d \tau} J = \partial^2_s J + R(J,\partial_s c) \, \partial_s c,
    \end{equation}
    where $\partial^2_s$ is the second partial derivative with respect to $s$ of the components of $J(s,\tau)$.
\end{corollary}

\begin{proof}
    Since the manifold $\mathcal{M}$ is smooth, then the tangent space $T_{c(s,\tau)} \mathcal{M}$ along curve $c(s,\tau)$ is guaranteed to exist and vary smoothly. 
    For a given $\tau = \tau'$, let $\{e_1,\dots,e_n\}$ be an orthonormal basis for $s = 0$ where $e_1$ is the unit vector in the direction of $\partial_s c(0,\tau')$.
    We can define orthonormal parallel vector fields $E_1(s,\tau'),\dots,E_n(s,\tau')$ by parallel transporting the basis $\{e_1,\dots ,e_n\}$ along $c(s,\tau')$. 
    By choosing a parallel orthonormal basis, we have, by definition, $D E_i(s,\tau') / d s = 0$.   
    For $J(s,\tau') = \sum_i a_i(s,\tau') E_i(s,\tau')$, this leads to
    \begin{equation*}
        \begin{aligned}
            \frac{D}{ds}J & = \sum^n_{i=1} \frac{D}{ds}\left(a_i(s,\tau') E_i(s,\tau')\right) \\
            & = \sum^n_{i=1} \left( \partial_s a_i(s,\tau') E_i(s,\tau') + a_i(s,\tau') \frac{D}{ds} E_i(s,\tau') \right) \\ 
            & = \sum^n_{i=1} \partial_s a_i(s,\tau') E_i(s,\tau')
        \end{aligned}
    \end{equation*}
    which is just $\partial_s J$. It trivially follows $D^2 J/ds^2 = \partial_s^2 J$. \qedhere
\end{proof}
 
Since the covariant derivative is linear and the Riemannian curvature tensor is trilinear, i.e., linear in all three of its arguments \cite{do1992riemannian}, \cref{eq:jacobi_heat} is a \emph{linear parabolic PDE} for the vector field $J(s,\tau) = \partial_\tau c$. 
Due to its similarity with the nonhomogeneous heat equation and Jacobi fields in Riemannian geometry \cite{do1992riemannian}, we call \cref{eq:jacobi_heat} the \emph{Jacobi heat flow equation}.

Characterizing the stability of \cref{eq:jacobi_heat} is not immediate because it is, in general, a coupled system of PDEs with temporally- and spatially-varying coefficients due to the presence of the Riemannian curvature tensor $R$. 
The analysis of \cref{eq:jacobi_heat} becomes simpler by replacing the Riemannian curvature tensor with the so-called sectional curvature. 
The following lemma states the relationship between these two mathematical objects.

\begin{lemma}
    \label{lemma:sectional}
    Let $K : T \mathcal{M} \times T \mathcal{M} \rightarrow \mathbb{R}$ be the \emph{sectional curvature} of the Riemannian manifold $(\mathcal{M},g)$ given by \cite{do1992riemannian}
    \begin{equation}
        \label{eq:sectional}
        K(J,\partial_s c) = \frac{\langle J, R(J,\partial_s c) \partial_s c \rangle_g}{\langle J, J \rangle_g \langle \partial_s c, \partial_s c \rangle_g - \langle J, \partial_s c \rangle^2_g},
    \end{equation}
    where $J, \partial_s c \in T \mathcal{M}$ and  $R: T \mathcal{M} \times T \mathcal{M} \times T \mathcal{M} \rightarrow T \mathcal{M}$ is the Riemannian curvature tensor.
    If the sectional curvature is positive (resp. non-positive), then the inner product $\langle J, R(J,\partial_s c) \partial_s c \rangle_g$ is positive (resp. non-positive).
\end{lemma}
\begin{proof}
    By Cauchy-Schwarz, if $J(s,\tau)$ and $\partial_s c$ are linearly independent, then $\langle J, \, \partial_s c \rangle^2_g < \langle J, J \rangle_g \langle \partial_s c, \partial_s c \rangle_g$, so the denominator in \cref{eq:sectional} is always positive.
    If $K(J,\partial_s c) > 0$, then $ 0 \leq \langle J, R(J,\partial_s c) \partial_s c \rangle_g \leq K(J,\partial_s c) \langle J, J, \rangle_g \langle \partial_s c, \partial_s c\rangle_g$.
    If $K(J,\partial_s c) \leq 0$ then $\langle J, R(J,\partial_s c) \partial_s c \rangle_g \leq 0$. 
\end{proof}

The sectional curvature is a more intuitive characterization of the curvature of manifolds.
In three or fewer dimensions it is identical to the Gaussian curvature for a surface. 
In higher dimensions, it can be interpreted as the Gaussian curvature of a surface defined by a two-dimensional plane formed by any two linearly independent vector fields in $T \mathcal{M}$ \cite{do1992riemannian}. 
Having non-positive sectional curvature means $(\mathcal{M},g)$ exhibits properties analogous to a hyperbolic or flat surface, while having positive sectional curvature corresponds to $(\mathcal{M},g)$ exhibiting properties analogous to a sphere.

We state one last lemma before proceeding to the first main result of this work; the proof can be found in the \nameref{sec:appendix}.

\begin{lemma}[\textbf{Poincar\'e Inequality}]
    \label{lemma:poincare}
    Let $W : [0,1] \times  \mathbb{R}_+ \rightarrow T \mathcal{M} : (s,\tau) \mapsto W(s,\tau)$ be a smooth vector field. 
    If $W(0,\tau) = W(1,\tau) = 0$ , then
    \begin{equation*}
        \int_0^1 \langle W, W \rangle \, ds \leq \frac{1}{4} \int_0^1 \langle \partial_s W, \partial_s W \rangle \, ds.
    \end{equation*}
\end{lemma}

We are now ready to state our first main result.

\begin{theorem}
    \label{thm:main_1}
    The Jacobi heat flow equation \cref{eq:jacobi_heat} is exponentially stable in $L_2$ if the curvature for the Riemannian manifold $(\mathcal{M},g)$ satisfies $\langle J, R(J,\partial_s c) \partial_s c \rangle_g < 4 \langle J, J \rangle$ uniformly.
\end{theorem}

\begin{proof}
    Consider the Lyapunov functional
    \begin{equation*}
        V(J(\tau)) = \frac{1}{2 \alpha } \int_0^1 \langle J, J \rangle \, ds,
    \end{equation*}
    where we use the slight abuse of notation of expressing $J(s,\tau)$ as $J(\tau)$ on the left-hand side because the integral essentially eliminates the $s$ dependency.
    Differentiation and using \cref{eq:jacobi_heat} yields,
    \begin{equation*}
        \begin{aligned}
            \frac{d}{d \tau} V(J(\tau)) & = \frac{1}{\alpha} \int_0^1 \langle J, \frac{D}{d\tau} J \rangle \, ds \\
            & = \int_0^1 \langle J, \partial^2_s J \rangle \, ds + \int_0^1 \langle J, R(J,\partial_s c) \partial_s c \rangle \, ds.
        \end{aligned}
    \end{equation*}
    Using integration by parts for the first integral, we obtain
    \begin{equation*}
        \begin{aligned}
            \frac{d}{d \tau} V(J&(\tau)) = \langle J, \partial_s J \rangle |^{s=1}_{s=0} - \int_0^1 \langle \partial_s J, \partial_s J \rangle \, ds \\ 
            & \hspace{1.cm} + \int_0^1 \langle J, R(J,\partial_s c) \partial_s c \rangle \, ds \\ 
            & = - \int_0^1 \langle \partial_s J, \partial_s J \rangle \, ds + \int_0^1 \langle J, R(J,\partial_s c) \partial_s c \rangle \, ds,
        \end{aligned}
    \end{equation*}
    where the evaluation of the inner product is zero because $J(0,\tau) = J(1,\tau) = 0$. 
    Using \cref{lemma:poincare}, 
    \begin{equation*}
        \begin{aligned}
            \frac{d}{d \tau} V(J(\tau)) & \leq  - 4 \int_0^1 \langle J, J \rangle \, ds + \int_0^1 \langle J, R(J,\partial_s c) \partial_s c \rangle \, ds, \\ 
            & = - 8 \, \alpha \, V(J(\tau)) + \int_0^1 \langle J, R(J,\partial_s c) \partial_s c \rangle \, ds.
        \end{aligned}
    \end{equation*}
    The integral term in the above expression can be interpreted as a perturbation to the nominal exponential rate $8 \alpha$ when the curvature of $(\mathcal{M},g)$ is non-zero, i.e., $R(J,\partial_s c) \partial_s c \neq 0$.
    The most straightforward way to analyze the integral term is to use to notion of sectional curvature introduced in \cref{lemma:sectional}.    
    If $(\mathcal{M},g)$ has non-positive sectional curvature $K(J,\partial_s c) \leq 0$,
    then $\langle J, R(J,\partial_s c) \partial_s c \rangle < 0$ so $d V (J(\tau)) / d \tau \leq - 8\, \alpha \, V(J(\tau)) \implies V(J(\tau)) \leq V(J(0)) e^{-8\alpha\tau}$ and $V(J(\tau))$ converges to zero exponentially.
    If instead $(\mathcal{M},g)$ has positive sectional curvature, 
    $\langle J, R(J,\partial_s c) \partial_s c \rangle_g \leq K(J,\partial_s c) \langle \partial_s c, \partial_s c\rangle_g \langle J,J \rangle_g < \kappa \langle J,J \rangle$ for $\kappa \in \mathbb{R}_{>0}$, then $d V(J(\tau)) / d \tau \leq -2 \alpha (4 - \kappa) V(J(\tau))$ so $V(J(\tau))$ still exhibits exponential convergence if $\kappa < 4$.    
    We can conclude $V(J(\tau)) \rightarrow 0$ exponentially so long as the curvature of $(\mathcal{M},g)$ is not too positive.
    Moreover, based on how $V(J(\tau))$ is defined, the Jacobi heat flow equation \cref{eq:jacobi_heat} is exponentially stable in $L_2$.
\end{proof}

We now state the second main result of this work.

\begin{theorem}
    \label{thm:main_2}
    Any curve that satisfies the geometric heat flow PDE \cref{eq:heat_flow} will converge exponentially in $L_2$ to a geodesic connecting any two points on manifold $\mathcal{M}$ if the curvature of ($\mathcal{M},g$) satisfies $\langle J, R(J,\partial_s c) \partial_s c \rangle_g < 4 \langle J, J \rangle$ uniformly.
\end{theorem}
\begin{proof}
    From \cref{prop:heat_flow} and \cref{eq:heat_flow} we have
    \begin{equation*}
       \frac{d}{d \tau} \mathcal{E}(c(\tau)) = - \frac{1}{\alpha} \int_0^1 \langle \partial_\tau c, \partial_\tau c \rangle_{g(c)} ds.
    \end{equation*}
    Integrating both sides from $\tau$ to $\infty$ gives
    \begin{equation*}
        \mathcal{E}(\gamma) - \mathcal{E}(c(\tau)) = - \int_\tau^\infty \left( \frac{1}{\alpha} \int_0^1 \langle \partial_\tau c, \partial_\tau c \rangle_{g(c)} ds  \right) d\tau,
    \end{equation*}
    where we have used the result of \cref{prop:heat_flow} to assert $\lim_{\tau\rightarrow \infty} \mathcal{E}(c(\tau)) = \mathcal{E}(\gamma)$.
    If we assume that the Riemannian metric tensor $G$ is uniformly bounded, then there exists $\eta \in \mathbb{R}_{>0}$ such that $G \preceq \eta I$.
    With this assumption, we arrive at
    \begin{equation*}
        \begin{aligned}
            \mathcal{E}(c(\tau)) - \mathcal{E}(\gamma) \leq 2 \, \eta \int_\tau^\infty \left( \frac{1}{2 \alpha} \int_0^1 \langle \partial_\tau c, \partial_\tau c \rangle \, ds  \right) d\tau,
        \end{aligned}
    \end{equation*}
    where the integral in parentheses is the Lyapunov functional from \cref{thm:main_1}.
    If we assume that the curvature of $(\mathcal{M},g)$ satisfies the bound stated in \cref{thm:main_1}, then there exists $\rho \in \mathbb{R}_{>0}$ such that $V(J(\tau)) \leq V(J(0))e^{-\rho \tau}$. 
    We then have
    \begin{equation*}
        \begin{aligned}
            \mathcal{E}(c(\tau)) - \mathcal{E}(\gamma) & \leq 2 \, \eta \int_\tau^\infty V(J(\sigma)) \, d \sigma \\
            & = 2 \, \eta \, V(J(0)) \int_\tau^\infty e^{-\rho \sigma} d \sigma \\
            & = C e^{-\rho \tau},
        \end{aligned}
    \end{equation*}
    for some $C \in \mathbb{R}_{>0}$, which implies that $\mathcal{E}(c(\tau)) \rightarrow \mathcal{E}(\gamma)$ exponentially as $\tau \rightarrow \infty$.
    Because of how the Riemannian energy is defined, and since it is positive definite for all non-trivial curves, then $c(s,\tau) \rightarrow \gamma(s)$ exponentially in $L_2$ as $\tau \rightarrow \infty$ if the curvature condition of \cref{thm:main_1} is met. \qedhere
\end{proof}

\section{PDE Pseudospectral Solver}
\label{sec:implementation}
The appeal of using the geometric heat flow equation \cref{eq:heat_flow} to compute geodesics is that i) any numerically stable solver is guaranteed to converge to a geodesic based on the results of \cref{sec:main_result} and ii) there are several fast algorithms (some parallelizable) that can solve \cref{eq:heat_flow}. 
If $c(s,\tau) = (x_1(s,\tau),\dots, x_n(s,\tau))$, 
then in local coordinates, \cref{eq:heat_flow} is
\begin{equation}
    \label{eq:coordiantes}
    \frac{1}{\alpha} \partial_\tau x_i = \partial^2_s x_i + \sum^n_{j,k=1} \Gamma^i_{jk} \, \partial_s x_j \, \partial_s x_k,
\end{equation}
where $\Gamma^i_{jk}$ are the Christoffel symbols given by
\begin{equation*}
    \Gamma^i_{jk} = \frac{1}{2} \sum_{m=1}^n g_{im}^{-1} \left( \frac{\partial g_{mj}}{\partial x_k} + \frac{\partial g_{mk}}{\partial x_j} - \frac{\partial g_{jk}}{\partial x_m} \right),
\end{equation*}
with $g_{ij}$ being the $i, j$ component of the metric tensor $G$ in local coordinates.
Solving \cref{eq:coordiantes} can be done efficiently with pseudospectral methods and the Chebyshev polynomials \cite{boyd2001chebyshev}
\begin{align*}
    T_0(z)=1, ~~ T_1(z) = z, ~~ T_{n+1}(z) = 2\,z\,T_{n}(z)-T_{n-1}(z),
\end{align*}
as basis functions for $c(s,\tau)$ where $z\in[-1,1]$.
In other words, we approximate each component $x_i(s,\tau)$ of the curve $c(s,\tau)$ as
$x_i(s,\tau) = \sum_{j=0}^D c_{ij}(\tau) T_j(s)$ where $D$ is the polynomial degree of the approximation and $c_{ij}(\tau)$ are the Chebyshev coefficients that must be determined as $\tau \rightarrow \infty$.

Rather than reformulating \cref{eq:coordiantes} as an ODE for the coefficients $c_{ij}(\tau)$, we generated $D+1$ collocation points for an arbitrary initial guess for the coefficients, and propagated the value of $c(s,\tau)$ at each node until convergence.  
We use the Chebyshev–Gauss–Lobatto (CGL) points
\begin{equation}
    \label{eq:cgl}
    s_k= \frac12 \left(1-\cos\left(\frac{k\pi}{D}\right) \right), \ \ k=0,1, \dots, D
\end{equation}
for $s=(z+1)/2\in[0,1]$, and replacing the right-hand side of \cref{eq:coordiantes} with the Chebyshev differentiation matrices \cite{boyd2001chebyshev}, the values of $x_i(s,\tau)$ at each node $s_k$ can be found numerically by recursively calling a numerical ODE solver
\begin{equation*}
    x_i(s_k,\tau_{\ell+1}) = \texttt{odeSolver} \Bigl( \partial_\tau x_i(s_k,\tau_\ell), \, \Delta \tau \Bigr),
\end{equation*}
where $\partial_\tau x_i({s_k},\tau_\ell)$ is \cref{eq:coordiantes} evaluated at $s_k$ (generated via \cref{eq:cgl}) and the current time step $\tau_\ell$ and $\Delta \tau$ is the step size of the ODE solver.
Once each $x_i(s_k,\cdot)$ changes less than a user-specified threshold, the nodes are projected back to a Chebyshev basis using the Vandermonde matrix to obtain a geodesic expressed with the Chebyshev basis \cite{boyd2001chebyshev}.
We found that only a small number of nodes was required to accurately approximate the right-hand of \cref{eq:coordiantes} for most cases we tested, which is likely due to the use of Chebyshev differentiation matrices, rather than other numerical differentiation techniques.
The method outlined above is essentially the method of lines \cite{schiesser2012numerical}, but with the spatial derivatives computed using pseudospectral differentiation.  

\section{Numerical Evaluations}
\label{sec:example}
The pseudospectral geometric heat flow method was evaluated in two different test cases.
The first was to compute a geodesic connecting two points on surfaces with known analytic Riemannian metrics, namely the sphere, torus, and egg box. 
This is analogous to solving the point-to-point motion planning problem where the path must live on a non-Euclidean surface.
The second was to generate geodesics for control, as is necessary with controllers designed using contraction theory \cite{manchester2017control}.
The high computational cost of calculating geodesics has been a major obstacle to using contraction-based controllers in practice, motivating the work by \cite{wang2020continuous,tsukamoto2021contraction}.
Our work immediately addresses this issue. 

Our PDE pseudospectral implementation used the CVODE solver from the ODES scikit Python library and was compared with the gradient descent method from \cite{leung2017nonlinear}.
The optimization-based method from \cite{leung2017nonlinear} minimizes the Riemannian energy functional for a curve represented with Chebyshev polynomials; the optimization problem is posed over the space of Chebyshev coefficients. 
The gradient descent method typically requires more sample points $N$ than theoretically needed, i.e., $N > D+1$, to improve integration accuracy. 
We reimplemented \cite{leung2017nonlinear} in Python to fairly compare the two methods.
Unless otherwise noted, we set $\alpha = 4$, and all other hyperparameters are explicitly stated in the tables or figures. 
All tests were conducted on a 2020 MacBook Pro with an 2GHz Intel Core i5 processor.

\subsection{Geodesics on 2D Surfaces}
The proposed PDE pseudospectral method was first tested by computing point-to-point geodesics on three 2D surfaces embedded in $\mathbb{R}^3$: the sphere, the torus, and the egg box. 
\cref{tab:geod-surfaces} shows the length and computation times computed by the PDE pseudospectral method and gradient descent for the three surfaces tested.
The first observation to note is that the geodesic lengths computed by each method are identical, which confirms the convergence analysis presented in \cref{sec:main_result}.
The second observation is that the PDE pseudospectral method exhibits much faster computation times than gradient descent for the sphere and torus.
This was attributed to the theoretical convergence guarantees of the proposed method (which gradient descent does not have) and the low computational overhead of propagating the geometric heat flow equation forward in $\tau$.
In contrast, the gradient descent method performed slightly better in computation time on the egg box surface.
This was due to having to use a high-degree Chebyshev polynomial basis to reasonably approximate the geodesic on such a complex surface.
It can thus be concluded that the PDE pseudospectral method works very well, both in terms of accuracy and computation time, but performance degradation can occur for complex surfaces, e.g., those with rapid oscillations in the curvature tensor.

\begin{table}[t]
\vspace{0.1in}
\centering
\footnotesize
\caption{Surface benchmarks for the PDE pseudospectral method and our gradient descent implementation of \cite{leung2017nonlinear}. $D$ is the degree of the polynomial used and $N$ is the collocation points for the gradient descent.}
\label{tab:geod-surfaces}
\renewcommand{\arraystretch}{1.25}
\setlength{\tabcolsep}{4pt}
\begin{tabular}{l|ccc|ccc}
\hline
\multirow{2}{*}{\textbf{Surface}} &
\multicolumn{3}{c|}{\textbf{PDE Pseudospectral}} &
\multicolumn{3}{c}{\textbf{Optimization \cite{leung2017nonlinear}}} \\
\cline{2-7}
 & $D$ & Length & Time (ms) & $D,N$ & Length & Time (ms) \\
\hline
Sphere &
$7$   & 2.33 & \textbf{6.63} & $7,11$   & 2.33 & 9.79 \\
Torus  &
$11$   & 16.5 & \textbf{5.04}  & $11,15$  & 16.5 & 20.2 \\
Egg Box &
$500$  & 7.36 & 150E3   & $250,350$& {7.36} & \textbf{130E3} \\
\hline
\end{tabular}
\end{table}

\begin{figure}[t]
    \centering
    \begin{subfigure}[t]{0.48\linewidth}
        \centering       
        \includegraphics[width=\linewidth,trim={15 30 20 30, clip}]{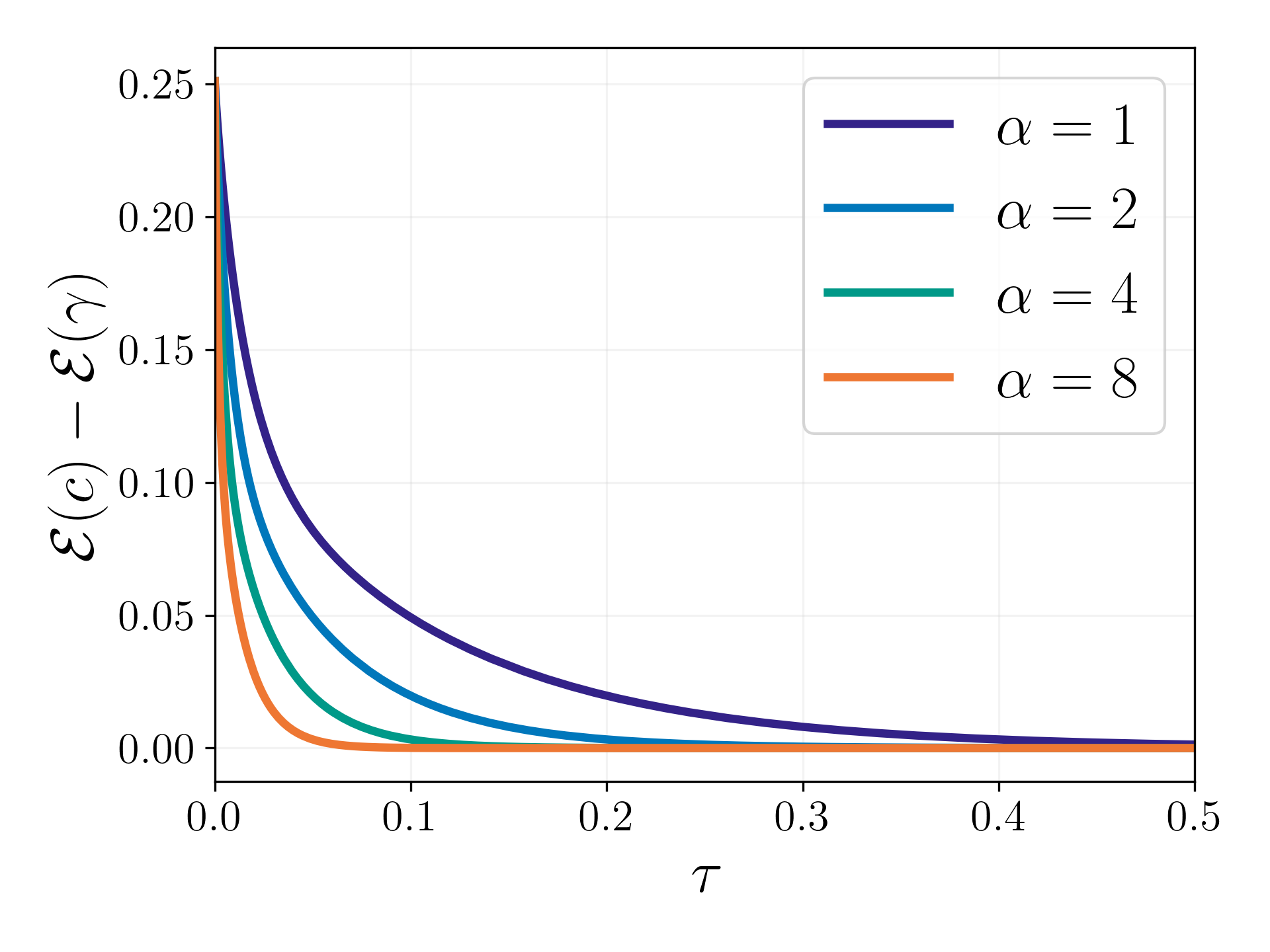}
        \caption{ Evolution of Riemannian energy on sphere for different values of $\alpha$.}
        \label{fig:alpha}
    \end{subfigure}
    \hfill
    \begin{subfigure}[t]{0.48\linewidth}
        \centering
        \includegraphics[width=\linewidth,trim={15 30 20 30, clip}]{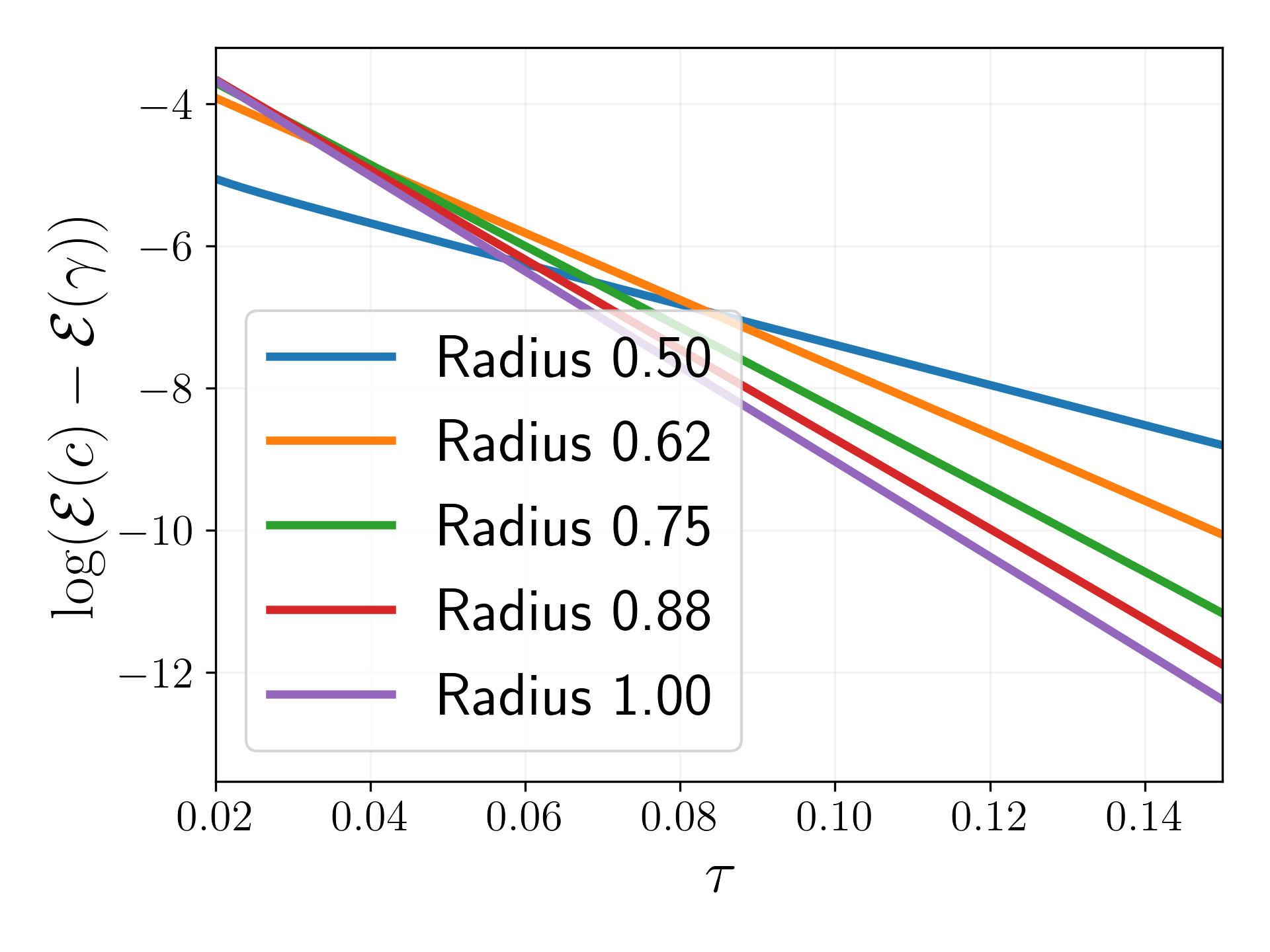}
        \caption{Convergence rate of Riemannian energy for spheres of different radii.}
        \label{fig:slopes}
    \end{subfigure}

    \caption{Convergence rate tests on spherical surface. (a): The proposed method has a high exponential convergence rate as $\alpha$ increases. (b): The exponential convergence rate of the proposed method is reduced as the radius of the sphere becomes smaller (curvature becomes larger).}
    \label{fig:sphere_res}
    \vspace{-0.25in}
\end{figure}

A more detailed study of the sphere was also conducted to confirm the effects of positive curvature on the convergence rate of the proposed method as predicted in \cref{sec:main_result}.
The baseline performance of the algorithm was first evaluated by quantifying how the hyperparameter $\alpha$ affects the convergence rate of $\cref{eq:heat_flow}$ for a sphere (which has constant positive sectional curvature). 
\cref{fig:alpha} shows that \cref{eq:heat_flow} indeed exhibits exponential convergence for a fixed radius sphere, and that the convergence rate increases with $\alpha$.
\cref{fig:slopes} shows that \cref{eq:heat_flow} still exhibits exponential convergence even as the radius of the sphere becomes smaller (sectional curvature increases), but the rate of convergence slows.
The behaviors observed in \cref{fig:sphere_res} all confirm the conclusions reached in \cref{sec:main_result}: curves propagated through the geometric heat flow equation \cref{eq:heat_flow} will exponentially converge to a geodesic with the convergence rate dictated by the hyperparameter $\alpha$ and the curvature of the Riemannian manifold.

\subsection{Contraction-Based Control with Geodesics in the Loop}
Contraction-based control leverages tools from Riemannian geometry to construct nonlinear tracking controllers (expressed implicitly) that exponentially track a desired trajectory \cite{manchester2017control}.
The underlining principle of contraction-based control is to find a Riemannian metric tensor offline that ensures that the Riemannian energy of a geodesic connecting the desired and actual state decreases exponentially. 
In other words, the Riemannian energy acts as a control Lyapunov function. 
The key online step is to compute a geodesic at each control epoch; as discussed above, this is computationally demanding if gradient descent is used. 
We compared the performance of our PDE pseudospectral solver with gradient descent for a third-order system that possesses a non-flat Riemannian metric \cite{manchester2017control,leung2017nonlinear,lopez2020adaptive}.
The metric tensor was found to be quadratic in one of the states via sum-of-squares programming. 
See \cite{manchester2017control,leung2017nonlinear,lopez2020adaptive} for more details about the system and the computed metric tensor.

\begin{table}[t]
\vspace{0.1in}
\centering
\caption{Comparison of computation times averaged over 100 runs for the example from \cite{leung2017nonlinear}. $D$ indicates the degree of the polynomial used, and $N$ is the validation nodes used for the optimization method.}
\label{tab:geodesic_results}
\renewcommand{\arraystretch}{1.2}
\setlength{\tabcolsep}{10pt}
\footnotesize
\begin{tabular}{c|cc|cc}
\hline
\multirow{2}{*}{\textbf{$x_0$}} &
\multicolumn{2}{c|}{\textbf{PDE Pseudospectral}} &
\multicolumn{2}{c}{\textbf{Optimization \cite{leung2017nonlinear}}} \\
\cline{2-3} \cline{4-5}
 & $D$ & Time (ms) & $D,N$ & Time (ms) \\
\hline
$[1,1,1]^\top$ & $4$  & \textbf{3.24} & $4,8$  & 5.34 \\
$[3,3,3]^\top$ & $4$  & \textbf{3.95} & $4,8$  & 7.81 \\
$[5,5,5]^\top$ & $5$  & \textbf{4.86} & $5,9$  & 10.4 \\
$[7,7,7]^\top$ & $6$  & \textbf{5.64} & $6,10$ & 17.3 \\
$[9,9,9]^\top$ & $7$  & \textbf{5.48} & $7,11$ & 23.0 \\
\hline
\end{tabular}
\vspace{-0.1in}
\end{table}

\begin{figure}[t]
\vspace{0.1in}
    \centering
    \begin{subfigure}[t]{0.48\linewidth}
        \centering
        \includegraphics[width=\linewidth,trim={10 30 20 30, clip}]{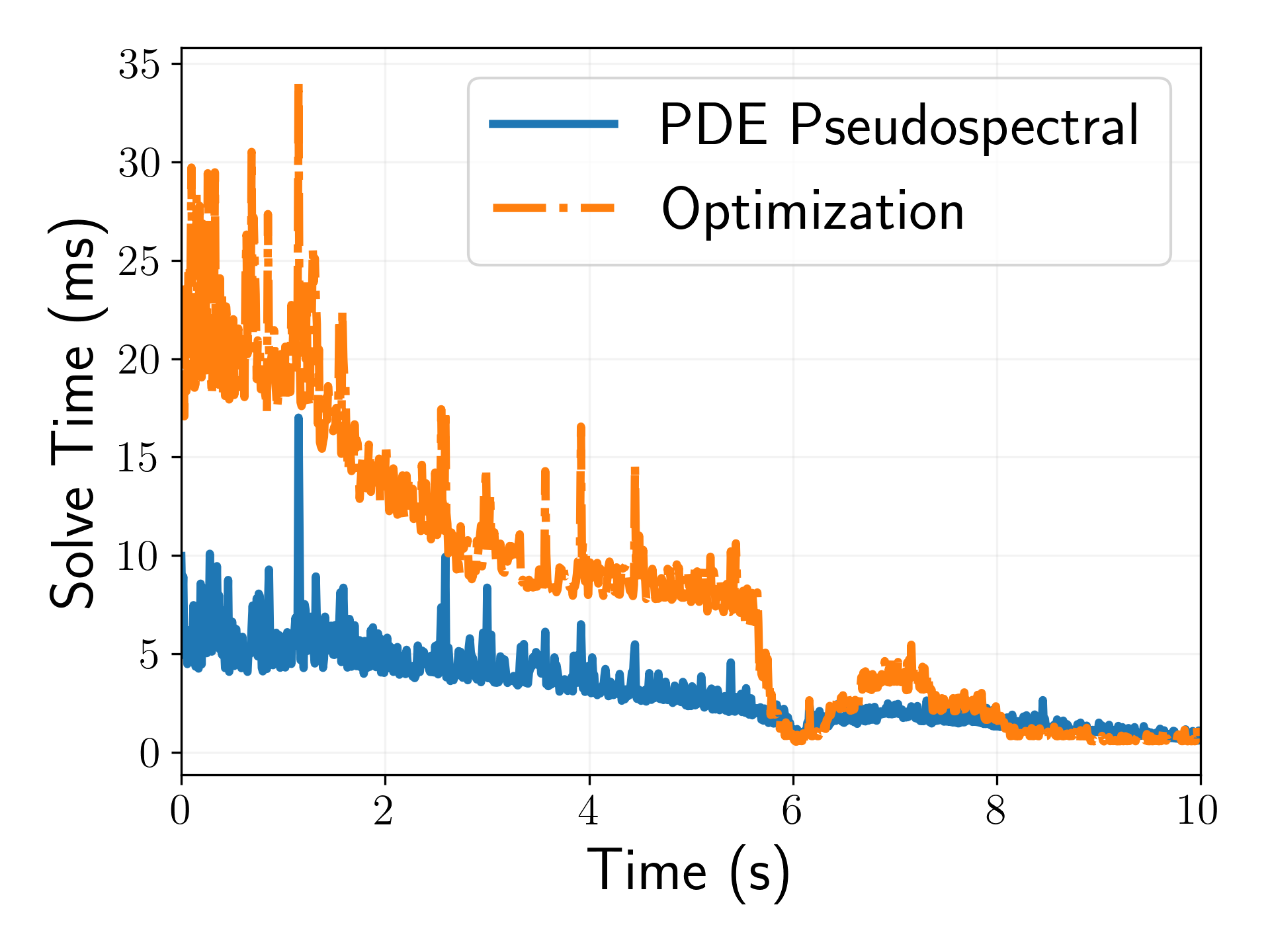}
        \caption{Time to compute the geodesic at each time step for both methods.}
        \label{fig:CCM_times}
    \end{subfigure}
    \hfill
    \begin{subfigure}[t]{0.48\linewidth}
        \centering        \includegraphics[width=\linewidth, trim={10 30 20 30, clip}]{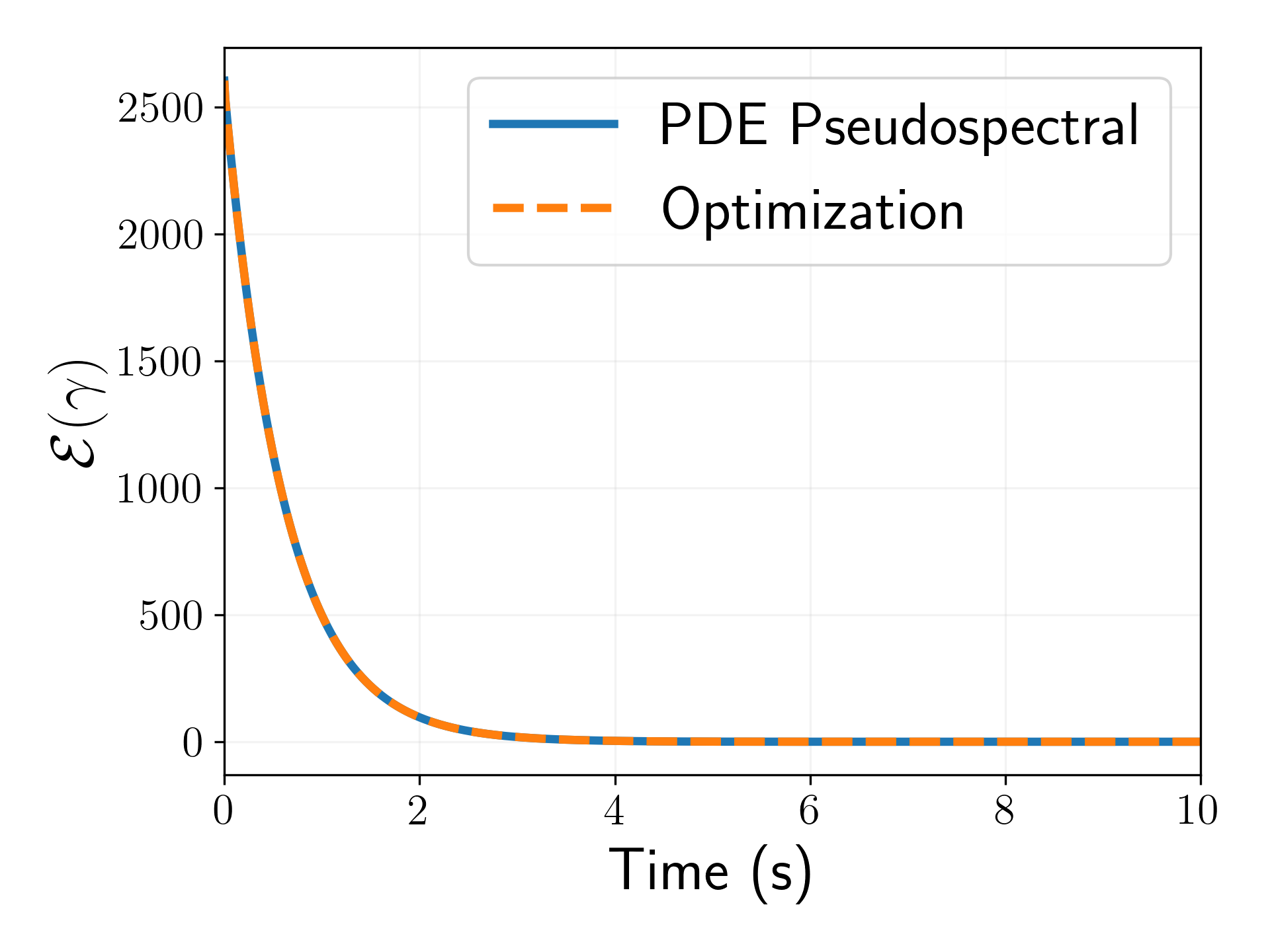}
        \caption{Evolution of Riemannian energies for both methods.}
        \label{fig:Energies_CCM}
    \end{subfigure}
    \hfill
    \caption{Comparison of the PDE pseudospectral method (ours) and gradient descent with a contraction-based controller. (a): The proposed method can be more than 3x faster than numerical optimization. (b) The Riemannian energy evolves nearly identically for both methods.}
    \label{fig:CCM_results}
    \vspace{-0.2in}
\end{figure}

The first set of tests we conducted was analogous to the previous case study: quantifying the computation time for the PDE pseudospectral and gradient descent methods. 
\cref{tab:geodesic_results} shows the computation times for the two methods where the initial state is varied (the desired state was the origin).
As seen in the table, the computation of gradient descent grows substantially---as much as 5x---as the initial state moves away from the origin.
Conversely, our PDE pseudospectral method sees a modest increase of at most 1.6x, again showcasing the method's computational efficiency.

We then compared our PDE pseudospectral method to the gradient descent method when the computed geodesics are used in feedback.
The metric and controller were constructed in the same way as in \cite{leung2017nonlinear,manchester2017control}.
The simulation was run with a timestep of $0.01$s, and a geodesic was calculated at each timestep to compute the control input.
We used $x_0=[9,9,9]^\top$ and $x_d=[0,0,0]^\top$, $D=7$ and $N=11$ for the simulation.
The time to compute the geodesic for the two methods is shown in \cref{fig:CCM_times}.
The PDE pseudospectral was almost 3x faster overall, especially when $x$ was further from $x_d$.
Additionally, \cref{fig:Energies_CCM} shows that the evolution of the Riemannian energy is identical for both methods, indicating that the accuracy of each method is comparable.
To summarize, \cref{fig:CCM_results} illustrates that the PDE pseudospectral method requires significantly less computation time compared to gradient descent, while maintaining comparable accuracy in terms of geodesic length.  

\section{Conclusion}
\label{sec:conclusion}
In this work, we presented an analysis of the geometric heat flow equation and a PDE pseudospectral solver that was demonstrated to be faster than computing geodesics using numerical optimization.
The key technical result is that a curve propagated through the geometric heat flow equation will exponentially converge in $L_2$ to a geodesic if the curvature of the Riemannian manifold does not exceed a positive bound and will otherwise asymptotically convergence.
Our theoretical analysis was validated through numerical tests that demonstrated exponential convergence and showed that the rate of convergence varies with the curvature.
Additional simulations were conducted to show that the proposed method is ideal for real-time feedback control using a contraction-based controller.
Future work includes investigating other PDE solvers for improved performance in different scenarios (e.g., egg box) and extending this method to Finsler manifolds.

\section*{Appendix}
\label{sec:appendix}

\begin{proof}[Proof of \cref{lemma:poincare}]
    Without loss of generality, let $w : [0,1] \times  \mathbb{R}_+ \rightarrow \mathbb{R} : (s,\tau) \mapsto w(s,\tau)$ be a component of a smooth vector field $W$.
    By the fundamental theorem of calculus, 
    \begin{equation*}
        w(s,\tau) - w(0,\tau) = \int_0^s \partial_\sigma w \, d \sigma \Rightarrow w(s,\tau) = \int_0^s \partial_\sigma w \, d \sigma.
    \end{equation*}
    Taking the absolute value and applying Cauchy-Schwarz, 
    \begin{equation*}
        \begin{aligned}
            |w(s,\tau)| & = \left| \int_0^s \partial_\sigma w \, d \sigma \right| \\ 
            &\leq  \left( \int_0^s d \sigma \right)^{1/2} \left(\int_0^s (\partial_\sigma w)^2 d \sigma\right)^{1/2}. 
        \end{aligned}
    \end{equation*}
    Squaring both sides and evaluating the second integral,
    \begin{equation*}
        w(s,\tau)^2 \leq s \int_0^s (\partial_\sigma w)^2 d \sigma \leq s \int_0^1 (\partial_\sigma w)^2 d \sigma,
    \end{equation*}
    where the second inequality arises because the integral of a positive semidefinite integrand is monotonic in the upper limit. 
     A second bound on $w(s,\tau)$ can be derived by applying the fundamental theorem but in the reverse direction, namely
    \begin{equation*}
        \begin{aligned}
            w(1,\tau) - w(s,\tau) = \int_s^1 \partial_\sigma w \, d \sigma \Rightarrow  w(s,\tau) = \int_1^s \partial_\sigma w \, d \sigma.
        \end{aligned}
    \end{equation*}
    Following identical steps as before yields
    \begin{equation*}
         w(s,\tau)^2 \leq (1-s) \int_0^1 (\partial_\sigma w)^2 d \sigma,
    \end{equation*}
    so we have the bound
    \begin{equation*}
        w(s,\tau)^2 \leq \min\{s,1-s\} \int_0^1 (\partial_\sigma w)^2 d \sigma. 
    \end{equation*}    
    Integrating both sides with respect to $s$,
    \begin{equation*}
        \int_0^1 w^2 \, ds \leq \int_0^1 \left( \int_0^1 (\partial_\sigma w)^2 d \sigma \right) \min\{s,1-s\} \, ds.
    \end{equation*}
    Noting that the inner integral is independent of $s$, then
    \begin{equation*}
        \begin{aligned}
            \int_0^1 w^2 \, ds & \leq \left( \int_0^1 (\partial_\sigma w)^2 d \sigma \right) \int_0^1 \min\{s,1-s\} \, ds \\
             & = \frac{1}{4} \int_0^1 (\partial_\sigma w)^2 d \sigma \\
             & = \frac{1}{4} \int_0^1 (\partial_s w)^2 d s,
        \end{aligned}
    \end{equation*}
    where we have replaced the dummy variable $\sigma$ with $s$.
    Since $w$ and $\partial_s w$ are arbitrary elements of $W$ and $\partial_s W$, the above inequality must hold for all components of $W$ and $\partial_s W$.
    Therefore, since $\langle W,W\rangle = \sum_i w^2_i$ and $\langle \partial_s W, \partial_s W\rangle = \sum_i (\partial_s w_i)^2$, we arrive at the desired inequality. Note that the $1/4$ factor is known as the Poincar\'e constant \cite{krstic2008boundary,hardy1952inequalities}.
\end{proof}

\subsection{Surface Metric Tensors and Parameters}
Parameters for the example surfaces and initial conditions are included here.
\vspace{-.2em}
\begin{table}[h!]
\caption{Sphere metric and boundary conditions.}
\centering
\begin{tabular}{|c|c|}
\hline
\multicolumn{2}{|c|}{\textbf{Sphere} ($R=1$)} \\
\hline
Coordinates & $\theta,\ \phi$ \\
\hline
Metric Tensor &  \rule[-2.3ex]{0pt}{5.7ex}$\begin{pmatrix} R^2 & 0 \\ 0 & R^2\sin^2\theta \end{pmatrix}$ \\
\hline
Start & $\pi/8,\ \pi/8$ \\
\hline
End & $3\pi/4,\ 2\pi/3$ \\
\hline
\end{tabular}
\label{tab:sphere}
\end{table}
\vspace{-.2em}
\begin{table}[h!]
\caption{Torus metric and boundary conditions.}
\centering
\begin{tabular}{|c|c|}
\hline
\multicolumn{2}{|c|}{\textbf{Torus} ($a=5,\ b=3$)} \\
\hline
Coordinates & $\theta,\ \phi$ \\
\hline
Metric Tensor &  \rule[-2.3ex]{0pt}{5.7ex} $\begin{pmatrix} (a+b\cos\phi)^2 & 0 \\ 0 & b^2 \end{pmatrix}$ \\
\hline
Start & $0,\ 0$ \\
\hline
End & $5\pi/4,\ 5\pi/4$ \\
\hline
\end{tabular}
\label{tab:torus}
\end{table}

\vspace{-.2em}
\begin{table}[h!]
\caption{Egg box surface and boundary conditions.}
\centering
\begin{tabular}{|c|c|}
\hline
\multicolumn{2}{|c|}{\textbf{Egg box}} \\
\hline
Coordinates & $x,\ y$ \\
\hline
Surface & $f(x,y)=x^2-y^2+2\sin(5x)\cos(5y)$ \\
\hline
Metric Tensor & \rule[-2.3ex]{0pt}{5.7ex} $\begin{pmatrix}1+f_x^2 & f_xf_y \\ f_xf_y & 1+f_y^2\end{pmatrix}$ \\
\hline
Start & $-1.5,\ -1.5$ \\
\hline
End & $1.5,\ 1.5$ \\
\hline
\end{tabular}
\label{tab:eggshell}
\end{table}

\bibliographystyle{IEEEtran}

\bibliography{references}

\end{document}